\documentclass[11pt]{article}
\usepackage{fullpage}
\usepackage{amssymb,amsmath}
\usepackage{graphicx, epsfig}
\usepackage{cite}
\usepackage{tikz}
\usetikzlibrary{fit,shapes, positioning}

\tikzset{
  treenode/.style = {align=center, inner sep=0pt, text centered,
    font=\sffamily},
  arn_n/.style = {treenode, circle, red, font=\sffamily\bfseries, draw=black,
     text width=1.5em, very thick},% real leaves
  arn_r/.style = {treenode, circle, dashed, black, draw=black,  
    text width=1.5em, very thick},% pseudo-leaves and pseudo-nodes
  arn_b/.style={treenode, circle, black, fill, draw=black, 
    text width=1.5em, very thick}, % real internal nodes
  arn_x/.style = {treenode, rectangle, white,
    minimum width=0.5em, minimum height=0.5em},% arbre rouge noir, nil
  level 1/.style={level distance=15mm,sibling distance=65mm},
  level 2/.style={level distance=15mm,sibling distance=32mm},
  level 3/.style={level distance=15mm,sibling distance=16mm},
  level 4/.style={level distance=15mm,sibling distance=8mm}
}

\def\floor#1{\left\lfloor #1 \right\rfloor}
\def\ceil#1{\left\lceil #1 \right\rceil}
\def\etal{\emph{et~al.}}

\newenvironment{itemize*}%
  {\begin{itemize}%
    \setlength{\itemsep}{0pt}%
    \setlength{\parskip}{0pt}%
    \setlength{\parsep}{0pt}%
    \setlength{\topsep}{0pt}%
    \setlength{\partopsep}{0pt}%
  }%
  {\end{itemize}}%

\newcommand{\eps}{\varepsilon}

\newcommand{\cC}{{\cal C}}
\newcommand{\Ts}{T^{\textsc{s}}}
\newcommand{\Tm}{T^{\textsc{m}}}

\newcommand{\consolid}{\textsc{Consolidate}}
\newcommand{\ow}{\overline{w}}
\newcommand{\oW}{\overline{W}}

\newtheorem{theorem}{Theorem}
\newtheorem{lemma}{Lemma}
\newtheorem{fact}{Fact}

\newenvironment{proof}{\trivlist\item[]\emph{Proof}:}%
{\unskip\nobreak\hskip 1em plus 1fil\nobreak$\Box$
\parfillskip=0pt%
\endtrivlist}

\begin{document}
\title{Dynamic  Trees with Almost-Optimal Access Cost}
\author{Mordecai Golin\thanks{Hong Kong University of Science and Technology. Email {\tt golin@cse.ust.hk}}
\and
John Iacono\thanks{Universit\'{e} libre de Bruxelles and New York University. Email {\tt johniacono@gmail.com}. Supported by NSF grants CCF-1319648, CCF-1533564, a Fulbright Fellowship, and by the Fonds de la Recherche Scientifique-FNRS under Grant no MISU F 6001 1.}
\and 
Stefan Langerman\thanks{Universit\'{e} libre de Bruxelles. Email {\tt sl@slef.org}. Directeur de recherches du Fonds de la Recherche Scientifique-FNRS.}
\and
J. Ian Munro\thanks{Cheriton School of Computer Science, University of Waterloo. Email {\tt imunro@uwaterloo.ca}}
\and
Yakov Nekrich\thanks{Cheriton School of Computer Science, University of Waterloo. Email {\tt yakov.nekrich@googlemail.com}}
}
\maketitle

\begin{abstract}
An optimal binary search tree for an access sequence on elements is a static tree that minimizes the total search cost.  Constructing perfectly optimal binary search trees is expensive so the most efficient algorithms construct  {\em almost optimal} search trees. There exists a long literature of constructing almost optimal search trees {\em dynamically}, i.e., when the access pattern is not known in advance.  All of these trees, e.g., splay trees and  treaps, provide a {\em multiplicative} approximation to the optimal search cost.

In this paper we show how to maintain an almost optimal  weighted binary search tree under access operations and insertions of new elements where the approximation is an {\em additive} constant. More technically, we maintain a tree in which the depth of the leaf holding an element $e_i$ does not exceed $\min(\log(W/w_i),\log n)+O(1)$ where $w_i$ is the number of times $e_i$ was accessed and $W$ is the total length of the access sequence. 

Our techniques can also be used to encode a sequence of $m$ symbols with a dynamic alphabetic code in $O(m)$ time so that the encoding length is bounded by $m(H+O(1))$, where $H$ is the  entropy of the sequence. This is the first efficient algorithm for adaptive alphabetic coding that runs in constant time per symbol. 
%improves upon the encoding time of the previous best algorithm (Gagie, 2004). 
\end{abstract}

\thispagestyle{empty}
\newpage

\section{Introduction}
\label{sec:intro}
The dictionary problem is one of the most fundamental problems in computer science. It requires maintaining 
a set of elements  in a data structure and being  able to efficiently  search for and find them when needed.  
In the comparison model, balanced binary search trees (BSTs) provide an optimal worst case solution for this problem. We consider leaf-oriented binary search trees, where all of the data is located in leaves and internal nodes store keys needed to guide the search to the leaves. For a set of $n$ elements, it is well known that the perfectly balanced search tree has height $\ceil{\log(n+1)}$ and  $\ceil{\log(n+1)}$ comparisons\footnote{Throughout this paper $\log$ denotes the binary logarithm and $\log^{(f)}$ is the $\log$ function iterated $f$ times.} are required to access an element, both in the worst and average cases.  In many practical applications,   some elements are known to be accessed  more frequently than others; unbalancing and restructuring the tree so that more frequently accessed elements are stored higher up,   can lead to better search times.  Let $d_i$ be the depth of the $i^{\mbox{\footnotesize th}}$  element  $e_i$ (stored at a leaf),  $w_i$ the frequency of accessing that element and $W=\sum_i w_i$ the total number of accesses. The total access cost is $\sum_i d_i w_i$; normalizing gives  the \emph{tree cost} which is  $ \frac 1 W \sum_i d_i  {w_i} .$  A tree that minimizes the tree cost minimizes the total access cost and  is an {\em optimal} BST.

There is a long literature on constructing optimal BSTs, both exactly and approximately\footnote{In this paper the term ``optimal'' refers to the optimality of the tree with respect to access frequencies. Splay trees, for example, can utilize other features of the access sequence in addition to frequencies.}.  In the approximate case, there are algorithms that provide both multiplicative and additive errors. 
In the dynamic version of the problem the frequencies $w_i$ are not known in advance but are calculated cumulatively as accesses are made.  The problem then is to update the tree to be optimal for the current observed frequencies.
Surprisingly,  while there are many results on dynamic  approximately optimal BSTs with constant multiplicative-error, prior to this paper there was not much known about   constant additive-errors.

In this paper we revisit this problem and describe  how to maintain dynamic  approximately optimal BSTs with constant {\em additive}-error (this will be formally defined in the next subsection).
The cost of re-building the tree after an access operation is bounded by $O(\log^{(f)}n)$ for any constant $f$, with the additive error growing linearly with $f.$  As in standard BSTs, our technique permits insertions of new elements to the dictionary at any time.

A variant of our approach can also be used to obtain an almost-optimal adaptive alphabetic code with  $O(1)$ encoding cost.

\paragraph*{Previous and Related Work.}
There are a number of data structures that maintain (unweighted) dynamic trees with $O(\log n)$ depth, starting with the classic balanced trees of Adelson-Velski and Landis~\cite{AVL62} and other handbook solutions~\cite{Bayer1972,GuibasS78}. These data structures maintain all leaves at height $O(\log n)$ and thus support both searches and updates, i.e., insertions and deletions,  in $O(\log n)$ time. 
The $k$-neighbor tree of Maurer et al.~\cite{MaurerOS76} achieves tree depth $(1+\delta)\log n$  and update cost $O((1/\delta)\log n)$ for any positive $\delta>0$. Andersson\cite{Andersson89} improved this result and showed how to maintain a tree of height $\log n + O(k)$ in $O(\log n)$ time per update. Even tighter bounds on constant and improved update times were described by Andersson and Lai~\cite{AnderssonL91} and Fagerberg~\cite{Fagerberg96}. We refer to~\cite{AnderssonFL04} for an extensive survey of results in this area.

Gilbert and Moore  \cite{Gilbert1959} introduced  an $O(n^3)$ time %dynamic programming 
algorithm for constructing optimal BSTs.  This was improved in 1971 by Knuth \cite{Knuth1971} to $O(n^2)$, which is still the best known method for solving the general case of the problem.  Those two algorithms assume that frequencies for both successful (elements in the tree) and {\em unsuccessful} (not in the tree) searches are given in advance and optimize accordingly.   If the problem is restricted to successful searches   then optimal BSTs can be constructed in  in $O(n \log n)$ time  using the Hu-Tucker algorithm and its variants.
\cite{Hu1971,Garsia1977}. Klawe and Mumey \cite{Klawe1995} show that, under some general conditions as to how the algorithms can operate,  $\Omega(n \log n)$ is the best possible construction time, although, for certain restricted types of input, $O(n)$ can be achieved \cite{Klawe1995,Hu2005}.

Let $p_i = w_i/W$ be the  empirical probability of element $i$ in the access sequence.  The {Shannon Entropy} of the sequence is $H =  \sum_i p_i \log (1/p_i)$ which   is known to be  a lower bound on the cost of  tree in which all data is in the leaves\footnote{When data can also be kept in internal nodes, as when three-way comparisons are allowed, the lower bound decreases to $H - \log H$~\cite{Allen82}.}.
If a tree was guaranteed to have $d_i \le  c +\log (1/p_i)$ for all $i$ then the total cost of all accesses would be  at most $\sum_i w_i (c+ \log (1/p_i) )= WH + cW$, i.e., 
within  a  constant additive error per access.
 In the static case multiple authors \cite{Ahlswede,Yeung1991,mehlhorn1977best} have provided $O(n)$ time algorithms for constructing such trees with $c=2$.
 
Now consider the  dynamic case, in which trees are rebuilt based on cumulative frequencies viewed so far.  {\em Splay  trees}  \cite{sleator1985self} and {\em Treaps,}  \cite{seidel1996randomized} maintain {\em static optimality},  essentially keeping  element $e_i$ at depth $d_i=O(\log (1/p_i))$ for the current cumulative frequencies, in the amortized sense.  This guarantees constant  {\em multiplicative} errors in the dynamic case.
There was no comparable result for maintaining almost optimal trees with additive errors, i.e., keeping element $e_i$ at depth $d_i=\log (1/p_i) +O(1).$   The best technique would be to rebuild the tree from scratch at every step.

The dynamic (or adaptive) alphabetic coding problem is closely related to the dynamic alphabetic tree problem just described. The coding problem is to produce an encoding for a sequence of symbols $S[1]\ldots S[m]$ over an ordered alphabet $\{\,a_1,\ldots, a_n\,\}$ so that (1) no codeword is a prefix of any other and (2) the codeword for $a_i$ is lexicographically smaller than the codeword for $a_j$ iff $a_i<a_j$. In the adaptive scenario the  input sequence is not known in advance; hence, we need to update the code every time a symbol is encoded. Dynamic Huffman \cite{Knuth85,Vit87}
and dynamic Shannon \cite{Gagie04} algorithms solve this problem for the non-alphabetic case.  The algorithm of Gagie~\cite{Gagie04} maintains a dynamic alphabetic code, such that the total encoding length is bounded by $(H+2)m$ and runs in $O(m(H+1))$ time.  

Alphabetic coding is related but not equivalent to the alphabetic trees problem.  Any alphabetic tree can be transformed into an alphabetic code in a straightforward way. Hence any dynamic alphabetic tree structure provides us with an alphabetic coding method.  But this imposes a lower bound on the encoding time: if the code is represented by a tree, then we have to encode the symbols bit-by-bit. Hence any  tree-based alphabetic coding method requires $\Omega(mH)$ time to encode the sequence. On the other hand, not every adaptive coding method can be transformed into a method for maintaining an alphabetic tree.  For example, the method of Gagie~\cite{Gagie04} does not store the alphabetic tree and therefore can not be used to implement a dynamic dictionary.

\paragraph*{Notation.} The \emph{weight} $w_{\ell}$ of a leaf node $\ell$ is the total number of times that an element stored in $\ell$ was accessed.  We assume that every item is accessed at least once so $w_\ell \ge 1$  The weight of an internal node $u$ is the total weight of all leaves in the subtree of $u$; the weight of a subtree is equal to the weight of its root.   The total weight $W$ of a tree $T$ is  the weight of its root node, i.e., $W= \sum_{\ell} w_{\ell}$ where the sum is taken over all leaves $\ell$.  This is also the total number of accesses made.

When necessary we further denote  by $w_{\ell}^{(j)}$  the number of accesses to $\ell$ during the first $j$ accesses.  Thus  $W^{(t)}= \sum_{\ell} w^{(t)}_{\ell}=t.$

\paragraph*{Relation Between Static and Dynamic Optimal Trees.} 

Consider an optimal static binary search tree for a sequence of $W$ accesses to $n$ elements. As previously noted, the  average cost of such a tree is at most $H+2$ where $H$ is the entropy of the access sequence.

\begin{lemma}
\label{lem:Dynamic Entropy}
Let $a_1,a_2,\ldots, a_W$ with $a_i \in \{1,2,\ldots,n\}$  be  a length $W$ {\em access sequence} on the elements, i.e., element $e_{a_t}$ is accessed at time $t$.
%To initialize, assume that immediately before starting, each item had been accessed exactly once to build the tree.
 Let $H$ be the entropy corresponding to the full access sequence.
% and set $p^{(t)}_i = \frac   {w^{(t)}_{i}} {W^{(t)}}$ to be the empirical probability based on the %first $t$ accesses. 
Then 
$$
%\sum_{t=1}^m  \log \frac 1 {p^{(t-1)}_{a_t}}  \le m H + 2m.
\sum_{t=1}^W  \log \frac t {\max\left(w^{(t-1)}_{a_t},\, 1\right)}  \le W\cdot H + 2W.
$$
\end{lemma}
The proof of this Lemma is straightforward and is therefore deferred to the Appendix.
% Section \ref{subsec:Lemma 1 proof}.

Suppose that we could build a tree $T^{(t)}$ such that the depth of $e_i$  {\em after} access $t$ is $d_i^{(t)} \le \log \frac t {w_i^{(t)}} +c$.
The access of $a_t$ at time $t$ would be in the previous tree $T^{(t-1)}$ with cost $d^{(t-1)}_{a_t}.$ The only exception to the above is if  time $t$ is the first access to $e_{a_t}$,  so it was not already in $T^{(t-1)}$. In that  case the access cost would  be  $d^{(t)}_{a_t}$, the depth of the location into which $a_i^t$ would be inserted.   Thus define $d^{(t-1)}_{a_t}=d^{(t)}_{a_t}$.  Since $w_{a_t}^{(t-1)}=0$ and $w_{a_t}^{(t)}=1$,
 $d_{a_t}^{(t-1)}= d_{a_t}^{(t)} \le \log  t+c = \log \frac t {\max\left(w^{(t-1)}_{a_t},\, 1\right)} + c.$
The total cost of the accesses would then, from  Lemma~\ref{lem:Dynamic Entropy}, be
$$\sum_i d^{(t-1)}_{a_t} \le W\cdot H + (2 +c) W,$$
i.e, within a constant additive error of optimal per access, where optimal defined as the cost with  the static optimal tree, is lower-bounded by $W\cdot H.$

Our approach to building almost optimal trees is therefore to maintain such trees $T^{(t)}$ over the access sequences.

\paragraph*{Our Results.}
Let $f \geq 1$ be any fixed integer. In this paper we describe a dynamic tree structure that can be maintained under access operations and insertions. The depth of the leaf that holds $e_i$ is bounded by $\min(\log(W/w_i),\log n)+ O(f)$ where $w_i$ is the number of times $e_i$ was accessed so far and $W$ is the total length of the access sequence.  Hence we can access any element $e_i$ using 
at most $\min(\log n,\log(W/w_i))+ O(1)$ comparisons. We can also insert new elements into the tree. When an element is accessed (resp.\ when a new element is inserted),  only $O(\log^{(f)}n)$ worst-case  time will be needed  to update the tree; this update procedure does not require any comparisons. Thus our data structure enjoys the advantages of both the weighted alphabetic  tree and the perfect binary tree. At the same time, the cost of maintaining the data structure is low. 

 This result is obtained by  a combination of two ideas. First, our construction is based on  approximate weights of elements instead of exact weights. Second, we maintain an unweighted binary tree $\Ts$ with leaves ``representing''  approximate weights. Our dynamic tree is  a subtree of $\Ts$. We define the approximate weights in Section~\ref{sec:approxim} and describe the tree $\Ts$ in Section~\ref{sec:static}. Next, we show how updates of our data structure can be implemented by leaf insertions in $\Ts$ in Section~\ref{sec:dynamic}. We reduce the update cost and make all time bounds worst-case in Sections~\ref{sec:updates} and~\ref{sec:worst} respectively. 

Our second result,  concerns the adaptive alphabetic coding problem. Our method enables us to encode the sequence of $m$ symbols with an adaptive alphabetic code in $O(m)$ time, constant time per symbol (in contrast to $O(m(H+1))$ time in \cite{Gagie04}). The length of encoding is bounded by $m(H+1)+ O(m)$ bits. Our  solution is based on the same approach as our dynamic tree structure, but we employ a different method to maintain the underlying tree. This method is based on the list maintenance problem~\cite{Willard92,BenderCDFZ02,BenderFGKM17}.  The full details of this result are omitted from this extended abstract but are presented  in the Appendix,  in Section~\ref{sec:alphabetic}.

\section{Preliminaries}
\label{sec:prelim}
An efficient solution for the unweighted search tree problem was presented by Maurer et al.~\cite{MaurerOS76}. Their data structure, called a {\em $k$-neighbor tree}, is a tree of height $(1+\delta)\log n$, where $\delta$ denotes an arbitrarily small positive constant.  
A  $k$-neighbor tree is a binary tree $T$ such that (1) all leaves in $T$ have the same depth and (2) if a node $u\in T$ has only one child, then $u$ has at least one right neighbor  (on the same level), and (3) if a node $u$ has $l$ right neighbors, then $\min(k,l)$ nearest right neighbors of $u$ have two children. 

Since this will be used later,  we give a sketch of the insertion into such a tree  below.

When a new leaf $x$  is inserted into the tree, we find the node $p$ such that the $x$ must be inserted below $p$  and call a recursive procedure $\textsc{Insert}(p,x)$. First, we make $x$ a new child of $p$. If $p$ has two children, the insertion  procedure is completed. If $p$ has three children, we look for a neighbor node $q$ of $p$ such that the distance between $p$ and $q$ is at most $k$ and $q$ has only one child. If $q$ is found, we call the procedure $\textsc{Move}(p,q)$. If $q$ is not found, we create a new node $p'$ that has  one child; the only child of $p'$ is the leftmost child of  $p$.  If  $p$ is not the root node, then we call the procedure $\textsc{Insert}(\text{parent}(p),p')$; otherwise we create a new root node $r_n$ and make both $p$ and $p'$ the children of $r_n$.
 
The arguments of the procedure $\textsc{Move}(p,q)$ are two neighbor nodes, $p$ and $q$, such that $p$ has three children and  $q$ has only one child. All nodes $u$ between $p$ and $q$ have two children. The procedure is applied to the children of all nodes $u$ between $p$ and $q$; every child node is shifted by one position to the right or to the left. At the end $p$, $q$, and all nodes $u$ have two children.  Thus $\textsc{Move}(p,q)$ consists of $d$ shifts, where $d$ is the distance from $p$ to $q$. Procedure $\textsc{Move}(p,q)$ needs $O(k)$ time because every node shift takes $O(1)$ time.  When a new leaf is inserted, we execute $\textsc{Move}(p,q)$ only one time. Excluding the cost of  $\textsc{Move}(p,q)$, we spend $O(1)$ time on every tree level. Therefore a new leaf can be inserted into a tree in $O(\log n+k)$ time.  
A more detailed description of an insertion can be found in~\cite{Andersson89}. We can delete a leaf using a symmetric procedure. 

The height of a $k$-neighbor tree with  $n$ leaves does not exceed $\floor{\frac{\log n}{\log(2-\frac{1}{k+1})}+1}$. 
Using the fact that for any  $k\ge\log n$ the height of the tree is bounded by $\log n + O(1)$, Andersson~\cite{Andersson89} showed how, by using an appropriate value of $k$ the tree height can be bounded by height $\log n + 2$ using only $O(\log n +k)= O(\log n)$ time per operation. It is this version of the data structure that we will use later.

\section{Approximate Weights}
\label{sec:approxim}
Consider an ordered weighted set of elements $E=\{\,e_1< e_2<\ldots < e_n\,\}$  let $w_i$ denote the  weight of $e_i$ and   $W=\sum_{j=1}^nw_j$.
Define the approximate (or quantized) weight of an element $e_i$ as  $w'_i=\ceil{w_i/\tau}$ for $\tau=\frac{W}{n}$. Thus all approximate weights are integers between $1$ and $n$. 
Note that $\sum\frac{w_i}{\tau}=\frac{n}{W} \sum_i w_i =n$. Hence $W'=\sum \ceil{\frac{w_i}{\tau}}\le \sum_i\frac{w_i}{\tau}+ n=2n\le 2W$.
%Since the weight of each element is incremented by less than $W/n$, the total approximate weight  $W'< 2W$. 

%Let $T$ be an optimal tree for weights $\{\,w_1,\ldots, w_n\,\}$ and let $T'$ denote an optimal tree for $\{\,w'_1,\ldots, w'_n\,\}$.
\begin{lemma}
  \label{lemma:approxim}
Suppose that the depth of a leaf $\ell_i$  in a tree $T'$ does not exceed $\log(W'/w'_i)+c$. Then the depth of $\ell_i$ in $T'$ does not exceed $\min(\log(W/w_i),\log n)+ c+1$.
\end{lemma}
\begin{proof}
   Since $w'_i\ge 1$ for all $i$, $\log (W'/w'_i)\le \log W'\le \log n +1$. 
Furthermore $W'\cdot \tau\le 2W$ and $w'_i\cdot \tau \ge w_i$. Hence $\frac{W'}{w'_i}=\frac{W'\cdot \tau}{w'_i\cdot \tau}\le \frac{2W}{w_i}$ and $\log\frac{W'}{w'_i}\le \log\frac{W}{w_i}+1$.

 %  Since $w'_i\ge W/n$ and $W'\le 2W$,  we have $W'\le  w'_i\cdot 2n$. Hence $\log(W'/w'_i)\le \log(2n)=\log n +1$.  
 % Since $W'\le 2W$ and $w'_i\ge w_i$, we have $\log(W'/w'_i)\le \log (2W/w_i)=\log(W/w_i)+1$.
 In summary  $\log\frac{W'}{w'_i}\le \min(\log\frac{W}{w_i},\log n)+1$.
\end{proof}
Thus the problem of maintaining an almost-optimal tree $T'$ for quantized weights $\{\,w'_1,\ldots, w'_n\,\}$ is equivalent to the problem of maintaining an almost-optimal tree for exact weights $\{\,w_1,\ldots, w_n\,\}$. The tree $T'$ has another important property: the depths of all leaves in $T'$ are bounded by $\ceil{\log n} + O(1)$. 
\begin{figure}[tbh]
  \centering
  \scalebox{0.55}{
  \begin{tikzpicture}[level/.style={sibling distance = 5cm/#1,
      level distance = 1.5cm}] 
    \node [arn_b] {}
    child{ node [arn_b] {} 
      child{ node [arn_b] {} 
        child{ node [arn_b] {} 
          child{ node [arn_n,label=left:$\eps_1$] {$e_1$}}
          child{ node [arn_r] {$e_1$}} 
        } 
        child{ node [arn_r] {}
          child{ node [arn_r] {$e_2$}}
          child{ node [arn_r] {$e_2$}} 
        }
      }
      child{ node [arn_b] {}
        child{ node [arn_n,label=left:$\eps_2$] {}
          child{ node [arn_r] {$e_2$}}
          child{ node [arn_r] {$e_2$}} 
        }
        child{ node [arn_r] {}
          child{ node [arn_r] {$e_3$}}
          child{ node [arn_r] {$e_3$}} 
        }
      }                            
    }
    child{ node [arn_b] {}
      child{ node [arn_n,label=left:$\eps_3$] {} 
        child{ node [arn_r] {}
          child{ node [arn_r] {$e_3$}}
          child{ node [arn_r] {$e_3$}} 
        }
        child{ node [arn_r] {}
          child{ node [arn_r] {$e_3$}}
          child{ node [arn_r] {$e_3$}} 
        }
      }
      child{ node [arn_b] {}
        child{ node [arn_r] {}
          child{ node [arn_r] {$e_3$}}
          child{ node [arn_r] {$e_3$}} 
        }
        child{ node [arn_b] {}
          child{ node [arn_n,label=above:$\eps_4$] {}}
          child{ node [arn_r] {$e_4$}} 
        }
      }
    }
    ; 
  \end{tikzpicture}}
\hspace{1pc}
  \scalebox{0.55}{
  \begin{tikzpicture}[level/.style={sibling distance = 5cm/#1,
      level distance = 1cm}] 
    \node [arn_b] {}
    child{ node [arn_b] {} 
      child{ node [arn_b] {} 
        child{ node [arn_b] {} 
          child{ node [arn_n] {$\eps_1$}}
        } 
      }
      child{ node [arn_b] {}
        child{ node [arn_n] {$\eps_2$}
        }
      }                            
    }
    child{ node [arn_b] {}
      child{ node [arn_n] {$\eps_3$} 
      }
      child{ node [arn_b] {}
        child{ node [arn_b] {}
          child{ node [arn_n] {$\eps_4$}}
        }
      }
    }
    ; 
  \end{tikzpicture}}

  \caption{Left: Balanced tree of approximate weights $w'_1=1$, $w'_2=2$, $w'_3=4$, and $w'_4=1$. Elements $e_1$, $\ldots$, $e_4$ are stored in nodes $\eps_1$, $\ldots$, $\eps_4$ respectively.  Pseudo-leaves are shown with dashed lines. Internal nodes of $\Ts$ that are not nodes of $T$ are also drawn with dashed lines. Leaves of $T$ are shown with solid lines and internal nodes of $T$ are depicted by filled circles. Right: Almost-optimal tree corresponding to the tree on Fig.~\ref{fig:static}}
  \label{fig:static}
\end{figure}
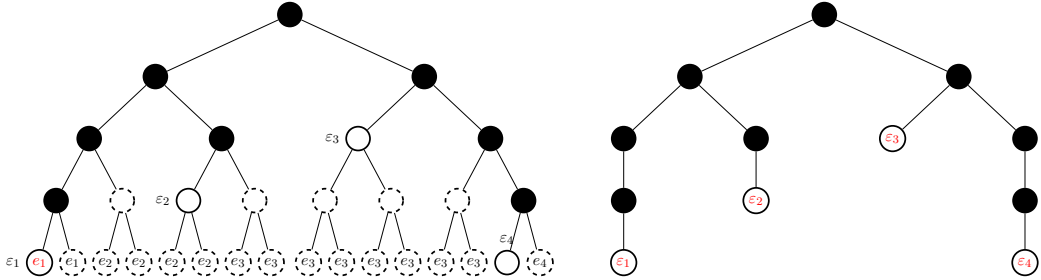

\section{Warm-Up: Almost-Optimal Static Trees}
\label{sec:static}
In this section we introduce our approach and  basic notions that will be used in the following sections. By way of introduction we describe a method that produces an almost-optimal tree  for a static set of elements with fixed weights.  

We keep weights of elements as entries in an array $B$ of size $m=2W' \le 2n$ so that there are two entries for each unit of weight. The  first $2w'_1$ entries of $B$ are assigned to $e_1$, the following $2w'_2$ entries are assigned to $e_2$, and so on. In general we assign entries $B[l_i]$, $\ldots$, $B[r_i]$ to the element $e_i$ where $l_i=(2\sum_{j=1}^{i-1}w'_i)+1$ and $r_i=2\sum_{j=1}^iw'_i$. Let $\Ts$ denote a conceptual perfectly balanced tree on $B$. The $i$-th leaf of $\Ts$ corresponds to the entry $B[i]$ of $B$, every internal node has two children, and the height of $T$ is $\log m=\log n +1$\footnote{To avoid tedious details, we assume in this section that $m$ and  $n$ are powers of $2$.}. %$\Ts$ provides an efficient way to explain our method, but it is not necessary to explicitly store $\Ts$. 
The leaves of $\Ts$ will be called \emph{pseudo-leaves}. Leaves corresponding to entries in $B[l_i..r_i]$ will be called  \emph{pseudo-leaves of the element $e_i$} (or pseudo-leaves associated to $e_i$).

\begin{fact}
  \label{fact:ancestor}
Consider a node $u$ of height $h\ge \floor{\log r}$ for some $r\ge 1$. Suppose that $r$ leftmost (rightmost) pseudo-leaves in the subtree of $u$ are pseudo-leaves of $e_i$. Then there is at least one node $v$ of height $\floor{\log r}$ such that all pseudo-leaves in the subtree of $v$ are pseudo-leaves of $e_i$. \\
Consequentially, if $2x$ entries are assigned to some element $e_i$, then there is at least one node $v$ of height $\floor{\log x}$, such that all pseudo-leaves in the subtree of $v$ are assigned to $e_i$. 
\end{fact}

We define an almost-optimal tree $T$ as  a subtree of $\Ts$. Let $\eps_i$ denote an arbitrary node of height $\floor{\log(w'_i)}$ such that all leaves in the subtree rooted at $\eps_i$ are $i$-nodes. Since we assigned $2w'_i$ pseudoleaves  to $e_i$, % it is easy to check that 
such a node $\eps_i$ always exists. All pseudoleaves below  $\eps_i$ correspond to some array entries in $B[l_i..r_i]$. The tree $T$ is a subtree of $\Ts$ pruned at nodes $\eps_i$. That is, the  nodes $\eps_i$ are the leaves of $T$ and all proper ancestors  of all $\eps_i$ are internal nodes of $T$.  We keep keys in the internal nodes of $T$ that can be used for routing. 

The depth of the leaf $\eps_i$ does not exceed $\log\frac{W'}{w'_i}$ by more than a constant: every leaf of $\Ts$ has depth at most $\log W'+ 1$. The depth of $\eps_i$ is then at most 
$$ \log(W')+1    - (\log(w'_i)+1) =\log\frac{W'}{w'_i}+ 2  \le \log\frac{W}{w_i}+ 3.$$

Hence each $\eps_i$ has an almost-optimal depth in $T$.  In addition $\Ts$ is a perfectly balanced tree with $2n$ nodes and the depth of any node in $\Ts$  does not exceed $\log n+1$.   Summing up, the depth of any leaf $\eps_i$ that holds the element $e_i$ does not exceed $\min(\log(W/w_i),\log n)+3$. 

An example tree $\Ts$ and the corresponding almost-optimal tree $T$ are shown on Fig.~\ref{fig:static}. An interesting property of our method is that the tree $T$ is not necessarily a full tree: it is possible that some internal nodes have only one child. In the following sections we will show how the tree $\Ts$ can be dynamized.

\section{Almost-Optimal Dynamic Trees}
\label{sec:dynamic}
Our dynamic data structure maintains a balanced  tree $\Ts$ on a dynamic set  $B$ of pseudo-leaves.

This first version  of the algorithm will work in phases.  
 A phase will end 
when the total weight $W$ is increased by a factor of $2$ or when the total number of elements is increased by a factor of $2$. 

Unlike in the previous section, these pseudo leaves are not kept in an array.  Instead,   $\Ts$ is maintained as a $k$-neighbor tree data structure with $k=\log n$~\cite{MaurerOS76} as described in Section \ref{sec:prelim}. 
  This method guarantees that all leaves of $\Ts$ have the same depth and , since the total number of pseudoleaves can at most double within a phase,  the height of the tree is bounded by $\log(4W') +1\le \log n+ 4$.  An update of $\Ts$ takes $O(\log^2 n)$ time. 

Each phase starts with a correct 
$\Ts$ that had just been built from scratch using the approach of Section~\ref{sec:static}.
Set  $\bar \tau= \tau = \frac{W}{n}$.  This value stays constant within the phase.

During  a phase, for every element $e_i$ we keep track of its weight $w_i$ and its approximate weight 
$w'_i=\ceil{w_i/ \bar \tau}$. Note that this implies that during  a phase
$w'_i$ can be increased (incremented by $1$ at a step)  but not decreased. 

 When $w'_i$ is incremented by $1$,
 %(i.e., after every sequence of $W/n$ accesses to $e_i$),  
 the tree $\Ts$ is updated:  we identify the rightmost pseudo-leaf $\ell_i$ associated to $e_i$ and insert two new pseudo-leaves, $\ell_n$ and $\ell_{n+1}$, immediately after $\ell_i$.  
 % If the value of  $h_i=\log(W'/w'_i)$ is changed, then we find the node $u'_i$ of height $h_i$ such that all pseudo-leaves below $u_i$ are associated to $e_i$ and assign  the element $e_i$ to $u_i$.
  When a new element $e_f$ is inserted into a tree, we insert two new pseudo-leaves, $\ell_f$ and $\ell_{f+1}$, into $\Ts$. The leaf $\ell_f$ is inserted after the leaf $\ell_p$, where $e_p$ is the largest element satisfying $e_p<e_f$ and $\ell_p$ is the rightmost leaf associated to $e_p$. Every insertion of a pseudo-leaf results in a modification of the tree $\Ts$.  

We maintain the almost-optimal tree $T$ as a subset of $\Ts$ using the approach of Section~\ref{sec:static}. An internal node $\eps_i$ is an internal node of $\Ts$ of height $\lfloor \log(w'_i)\rfloor$ such that all leaves in its subtree are associated to an element $e_i$. Using the same calculations as in 
Section~\ref{sec:static}
the depth of $\eps_i$ is then at most 
$ \log\frac{W}{w'_i}+ 4$ (and not $3$ because the calculation is using $\bar \tau$ and not $\tau$.)

 After an update of $\Ts$, some nodes of $\Ts$ (and, hence, some nodes of $T$) can be moved. If all leaves of a moved internal node $u$ are associated to $e_j$, we also update the internal node $\eps_j$, if necessary. Suppose that a node $u$ was moved by one position to the left and the node $u'$ to the right of $u$ was also moved by one position to the left. If all leaf descendants of $u$ are associated with an element $e_i$ and all leaf descendants of $u'$ are associated with some $e_j\not=e_i$, then we may have to update $\eps_i$. If $\eps_i$ is an ancestor of $u$, we find the immediate left neighbor $\eps'_i$ of $\eps_i$. Since there are $2w'_i$ leaves associated to $e_i$ and the height of $\eps_i$ is $\log(w'_i)$, all leaf descendants of $\eps'_i$ are associated to $e_i$. Hence we can set $\eps_i:=\eps'_i$. We can find the $\eps_i$ and $\eps'_i$ for every moved node $u$ in $O(\log n)$ time. At most $O(\log n)$ nodes of $\Ts$ are moved during every update \cite{Andersson89}; hence,  the total update cost is $O(\log^2 n)$.

When the total weight $W$ is increased by a factor $2$ or when the total number of elements is increased by a factor $2$, we update the value of $\tau=\frac{W}{n}$, compute the new values $w'_i$  and as noted, re-build the tree from scratch.  The amortized cost of rebuilding $\Ts$ from scratch 
is  $O(1)$ per step since the balanced tree can be built in linear time.  When we re-build the tree $\Ts$, we use the new value of $k=\log n$.
% Hence, the total amortized cost per step is $O(\log^2 n)$.

\begin{lemma}
\label{lemma:res1}

 We can implement a binary search tree so that  access to an element and an insertion of a new element  are supported in $O(\log^2 n)$ amortized time. If an element $e_i$ was accessed $w_i$ times over a sequence of $W$ operations, then the depth of the leaf holding  $e_i$  does not exceed $\min(\log(W/w_i), \log n)+O(1)$.
\end{lemma}

\section{Faster Updates}
\label{sec:updates}
We can reduce the update time by grouping pseudo-leaves in the tree $\Ts$.  All pseudo-leaves are divided into $\Theta(n/\log^2n)$ groups so that each group contains at least $\log^2 n$ and at most $2\log^2 n$ pseudo-leaves. 

\begin{figure}
\centerline{\includegraphics[width=4.0in]{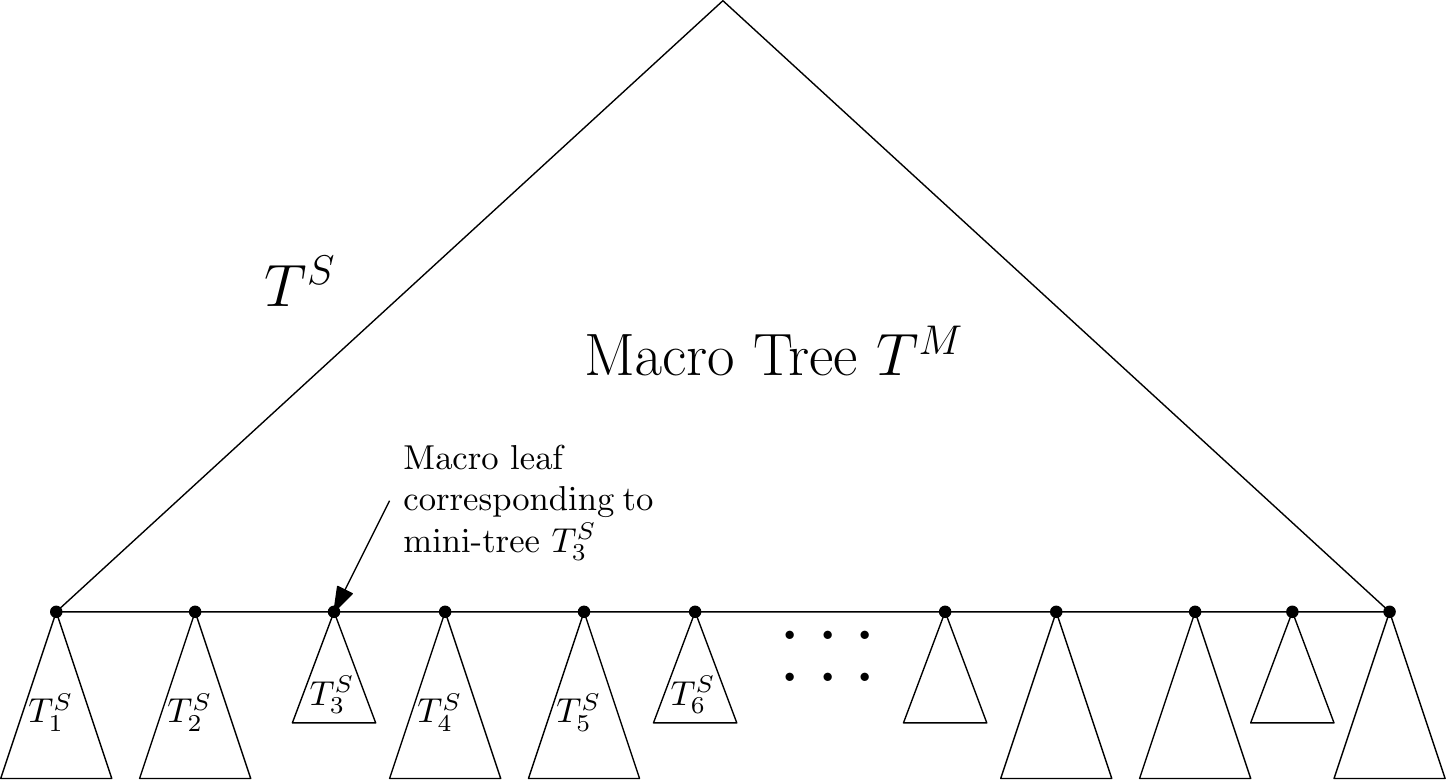}}
\caption{The partition of $T^S$ into macro tree $T^M$ and mini-trees $\Ts_j$.  The leaves of $T^m$ are the roots of the $\Ts_j$.  All the $\Ts_j$ have between $\log^2 n$ and $2\log^2 n$ pseudoleaves. 
% and at least one of every consecutive pair  $\Ts_j$,  $\Ts_{j+1}$ have $\Theta(\log^2 n)$ pseudo-leaves. 
 $T^M$ and all of the $\Ts_j$ are maintained as dynamic almost-optimal trees for their sets of leaves using the technique of Section \ref{sec:dynamic}.}
\label{fig:Grouping}
\end{figure}

The  tree $\Ts$ is divided into two components: a macro-tree  $T^M$with $O(n/\log^2 n)$ leaves and $O(n/\log^2 n)$ mini-trees $\Ts_j$.  See Fig.~\ref{fig:Grouping}. 
Mini-trees correspond to groups of pseudo-leaves: all pseudo-leaves in the group $G_j$ are stored in a mini-tree $\Ts_j$. 
The $j$'th leaf of  macro-tree $T^M$ is  the root of  mini-tree $\Ts_j$.  %Leaves of the macro-tree $\Tm$, further called macro-leaves, correspond to groups of pseudo-leaves.   
%The macro-leaf  corresponding to $G_j$ is the root of the tree $\Ts_j$. 
As before, the almost-optimal tree $T$ is a subtree of $\Ts$. An element $e_i$ is assigned to a node $\eps_i$ of  $T$, such that the height of $\eps_i$ in $\Ts$ is $\log (w'_i)$ (up to an additive  constant error) and all leaves in the subtree of  $\eps_i$ are associated to $e_i$.  $T$ is the subtree of $\Ts$ induced by nodes $\eps_i$ and their ancestors.
The  division of a tree into macro-trees and mini-trees is a standard data structuring technique; see e.g.,\cite{AnderssonL91}.

We now find the  node $\eps_i$ for any element $e_i$ either in a mini-tree or in the macro-tree. Recall that there are $2w'_i$ pseudo-leaves associated to $e_i$.  Let $g = 2 \log^2 n.$

First suppose that $w'_i\le g$; then the pseudo-leaves of $e_i$ are distributed among 
% a constant number of 
$O(1)$  subtrees. If all pseudo-leaves are in one subtree $\Ts_j$, then $\Ts_j$ has at least one node $u$ of height $\floor{\log(w'_i)}$ such that all leaves below $u$ are associated to $e_i$.  If pseudo-leaves of $e_i$ are in two subtrees, $\Ts_j$ and $\Ts_{j+1}$, then either $w'_i$ rightmost pseudo-leaves in $\Ts_j$ are associated to $e_i$ or $w'_i$ leftmost leaves in $\Ts_{j+1}$ are associated to $e_i$.  Hence either $\Ts_j$ or $\Ts_{j+1}$ contains a node that can be chosen as $\eps_i$. If pseudo-leaves of $e_i$ are distributed among more than two mini-trees, then there is at least one mini-tree $\Ts_j$ with all pseudo-leaves associated to $e_i$. In the latter case we can choose the root of $\Ts_j$ as $\eps_i$.

Now  suppose that $kg\le w'_i< (k+1)g$ for some $k\ge 1$. Then there are at least $2k-1$ mini-trees with all pseudo-leaves associated to $e_i$. The roots of these mini-trees are macro-leaves $\ell_j$, $\ldots$, $\ell_{j+2k}$. There is at least one node $u$  of height $\floor{\log k}$ in the macro-tree, such that all macro-leaves below $u$ are  among $\ell_j$, $\ldots$, $\ell_{j+2k}$.   

Using Lemma~\ref{lemma:res1}, we maintain the mini-tree $\Ts_j$  for every group $G_j$. Since each mini-tree has $O(\log^2 n)$ leaves, updates on a mini-tree take $O((\log\log n)^2)$ time.  The macro-tree is updated only when a new mini-tree is inserted or a mini-tree is deleted. Hence the cost of updating the macro-tree can be distributed among $O(\log^ 2n)$ insertions of pseudo-leaves. Suppose that a new pseudo-leaf corresponding to an element $e_i$ is inserted. As in Section~\ref{sec:dynamic} we find the rightmost pseudo-leaf $\ell'_i$ corresponding to an element $e_i$.  The new pseudo-leaf $\ell_i$ is inserted into the same mini-tree as $\ell'_i$  immediately to the right of $\ell'_i$. Since every mini-tree has $O(\log^2 n)$ pseudo-leaves, we can insert a new pseudo-leaf in $O((\log\log n)^2 )$ time. If the number of pseudo-leaves in $\Ts_j$ is equal to  $2\log^2 n$, we split the mini-tree $\Ts_j$ into two mini-trees of size $\log^2 n$; then we insert a new macro-leaf into $\Tm$. The cost of an insertion into $\Tm$ is $O(\log^2 n)$. We can also split a mini-tree into two mini-trees in $O(\log^2 n)$ time. Hence the amortized cost of maintaining the macro-tree is $O(1)$. 

The total height of a tree does not exceed the height of the macro-tree plus the maximum height of a mini-tree. Since the number of mini-trees is bounded by $\frac{2W'}{(\log^2 n)/2}$, the height of the macro-tree does not exceed $\log(W')-2\log\log n+ 3$.  The height of a mini-tree is bounded by $2\log\log n+1+O(1)$ because it contains at most $2\log^2n$ pseudo-leaves. 
Hence the total height of our tree does not exceed $\log(W')+ 4$. We already showed that the  height of a sub-tree rooted at the node $\eps_i$ is $\floor{\log(w_i')}$; hence the depth of $\eps_i$ in $T$ is at most $\log(W'/w'_i)+ O(1)$. 
\begin{lemma}
  \label{lemma:res2}
  We can implement a binary search tree so that  access to an element and an insertion of a new element  are supported in $O((\log \log n)^2)$ amortized time. If an element $e_i$ was accessed $w_i$ times over a sequence of $W$ operations, then the depth of the leaf holding  $e_i$  does not exceed $\min(\log(W/w_i), \log n)+O(1)$.
\end{lemma}

The result of Lemma~\ref{lemma:res2} can be further improved by bootstrapping.  For any integer $f\ge 1$ the following statement can be proved.  
\begin{lemma}
\label{lemma:recur}
  Suppose there exists a binary search tree $T^f$, such that (1) the depth of a leaf holding an element $e_i$ in $T^f$ does not exceed $\min(\log(W/w_i),\log n) +O(1) + O(f)$ (2) the amortized cost of updating $T^f$ after an element access or an insertion is $O((\log^{(f)}n)^2)$.\\
Then there is a binary search tree $T^{f+1}$, such that  (1) the depth of a leaf holding an element $e_i$ in $T^{f+1}$ does not exceed $\min(\log(W/w_i),\log n) +O(1) + O(f+1)$ (2) the amortized cost of updating $T^{f+1}$ after an element access or an insertion is $O((\log^{(f+1)}n)^2)$.
\end{lemma}
\begin{proof}
  We divide the tree $\Ts$ into the  macro-tree and mini-trees in the same way as in the proof of Lemma~\ref{lemma:res2}.  Every mini-tree is implemented using the tree $T^f$. Hence each mini-tree can be updated in $O((\log^{(f)} (\log n))^2)=O((\log^{(f+1)}n)^2)$ time. The amortized cost of maintaining the macro-tree is $O(1)$. Hence the total amortized cost of updates is $O((\log^{(f+1)}n)^2)$. 

  Suppose that $\eps_i$ is stored in the macro-tree. The depth of a node $\eps_i$ in the macro-tree is bounded by $\log(\min(W'/w'_i,n))+O(1)$. Now suppose that $\eps_i$ is stored in some mini-tree. The depth of $\eps_i$ in the mini-tree is bounded by $\log(\min(W'_g/w'_i,n_i))+O(f)+O(1)$, where $W'_g$ is the total sum of all quantized weights in the mini-tree and $n_g$ is the total number of elements in the subtree.  By the same argument as in Lemma~\ref{lemma:res2},  the depth of $\eps_i$ in $T$ is bounded by $\log(\min(W'/w'_i,n))+O(f+1)+O(1)$. 
\end{proof}

Our main result is obtained when we apply Lemma~\ref{lemma:recur} $f+1$ times for a parameter $f\ge 0$.
\begin{theorem}
  \label{theor:main}
  For any $f\ge 1$ there exists a binary search tree $T^{f}$, such that the depth of a leaf holding an element $e_i$ in $T^{f}$ does not exceed $\min(\log(W/w_i),\log n) +O(f)$ and the amortized cost of updating $T^{f}$ after an element access or an insertion is $O(\log^{(f)}n+f)$.
\end{theorem}
We remark that when we insert a new element $e_i$, we need to update the search path for one leaf. This may incur an additional cost of $\log n + O(1)$ operations.  

Our data structure can also support two symmetric operations. We can decrement the weight of an element and delete an element of weight $1$.  These operations can be implemented in the same way as incrementing the weight of an element and an insertion of a new element.  

\section{Worst-Case Updates}
\label{sec:worst}
Our construction can be modified to support updates with worst-case time guarantees. We start by showing how the data structure from Section~\ref{sec:dynamic} can be changed. We run several processes in the background; these processes adapt the tree structure to the changing value of the parameter $\tau$ and maintain the correct number of pseudo-leaves for each element $e_i$. The value of $\tau$ is changed every time  the total weight $W$ for the  number of elements is changed by a constant factor (described below). Two background processes guarantee that the value of $\tau$ used in $\Ts$ is within a constant factor of its current  value.  Moreover pseudo-leaves are stored in a $k$-neighbor tree data structure, but the parameter $k=\Theta(\log n)$  must  be changed when the number of elements is increased or decreased by too much. We run another process that modifies the tree when  the parameter $k$ needs to be changed.

Let $W_0$ and  $n_0$ denote the total weight and the number of elements at some time $t_0$. Let  $\tau_0=W_0/n_0$ and let the \emph{delayed weight} of an element $e_i$ be defined as $\ow_i=\ceil{w_i/\tau_0}$. We maintain the invariant that $w'_i$ differs from $\ow_i$ by at most a constant factor.  In the worst-case construction delayed weights $\ow_i$ are used instead of $w'_i$, i.e., an element $e_i$ is assigned $\ow_i$ pseudo-leaves.  Our re-building processes guarantee that $W_0\le W\le (4/3)W_0$ and $n_0\le n\le (4/3)n_0$. Therefore $\tau=(W/n)\le (4/3)\tau_0$ and $\tau\ge (3/4)\tau_0$.  For any element $e_i$, $\ow_i=\frac{w_i}{\tau_0}\le (4/3)w'_i$ and $\ow_i\ge (3/4)w'_i$. Thus we have $\frac{\oW}{\ow_i}\le (16/9)\frac{W'}{w'_i}$  where $\oW=\sum_i \ow_i$ and $\log(\oW/\ow_i)< \log(W'/w'_i)+ 1\le \log(W/w_i)+2$. 

We move among three re-building processes. Each process is executed in the background during at most $n/18$ insertions or accesses.  The first process updates the value of $W_0$. If $W\ge (7/6)W_0$, we  set $W_1=W$, $n_1=n$, and compute $\tau_1=W_1/n_1$. For every $e_i$, we compute the new value of $\ow_i=w_i/\tau_1$ and update the tree $\Ts$ by removing some pseudo-leaves if necessary. When the number of pseudo-leaves for all elements is adjusted in this way, we set $W_0=W_1$ and $n_0=n_1$. 
The second process updates the value of $n_0$. If $n\ge (7/6)n_0$, we also compute the new $\tau_1=W_1/n_1$ for $W_1=W$ and $n_1=n$. Then for every element $e_i$ we set $\ow_i=w_i/\tau_1$ and update the tree $\Ts$. The tree always contains $O(n)$ leaves. Every time when we access an element or insert a new element, our background process inserts or removes $O(1)$ pseudo-leaves. We can choose the constant in such a way that adjusting the value of $\tau_0$ is distributed among  $n/18$ update or access operations. Suppose that $W\ge (7/6)W_0$ or $n\ge (7/6)n_0$; the value of $\tau_0$ will be adjusted after at most $n/6$ operations. Hence $W\le (4/3)W_0$ and $n\le (4/3)n_0$ at any time.
  
The third background process updates the parameter $k$ in the $k$-neighbor tree. We set $k_0=2(\log (W_0)+1)$ and maintain a $k_0$-neighbor tree on pseudo-leaves. When the number of leaves in $\Ts$ is increased by factor $2$, we start the process of adjusting $k$. Internal nodes on every level of the tree are divided into pieces, so that every piece consists of $k_0+2$ consecutive nodes. We process pieces on the same level in the left-to-right order. Since $\Ts$ is already a $k_0$-neighbor tree, the distance between any two $1$-nodes (a $1$-node is a node with one child) is at least $k_0+1$. Hence each piece contains at most two $1$-nodes. If there are two $1$-nodes in the same piece $P$, then they are the leftmost and the rightmost nodes in $P$. In this case, we execute  the procedure  $\consolid(u,u')$, where $u$ and $u'$ are the $1$-nodes in $P$. This procedure, that will be described below,  removes the node $u$ and adds one additional child to  $u'$.  If $P$ contains one $1$-node, then we examine the preceding piece $P'$. If $P'$ also contains a $1$-node and the distance between the $1$-nodes in $P$ and $P'$ is equal to $k_0+2$, we start the procedure $\consolid(u',u)$, where $u$ is the $1$-node in $P$ and $u'$ is the $1$-node in the slide that precedes $P$ . After all pieces on a tree level are processed, every piece contains at most one $1$-node and the distance between $1$-nodes is at least $k+2$. We will show below that $\consolid$ requires  $O(k\log  n)$ move operations and can be executed in $O(k\log^2 n)=O(\log^3 n)$ time. Hence the third background process needs $O((n/k)\log^3 n)=O(n\log^2 n)$ time. Since an update takes $O(\log^2 n)$ time, we can distribute the third process among $n/18$ tree updates or accesses. 

\begin{figure}[tb]
\resizebox{\textwidth}{!}{%
  \centering
  \begin{tikzpicture}[every node/.style={circle,draw},level distance=8mm, %
    level 1/.style={sibling distance=8mm,level distance=8mm},
    level 2/.style={sibling distance=8mm,level distance=8mm},
    level 3/.style={sibling distance=8mm,level distance=8mm},
    level 4/.style={sibling distance=8mm,level distance=8mm},
    ]
    at (-6,3) \node (v) {$v$}
    child{ node{}
      child{ node (ou) {}
        child{ node (u) {$u$}
          child{ node{}
          }
        }
      }
    }
    child{ node{}
    };

    \node (u1) [right=1cm of u] {$u_1$}
    child{ node{}}
    child{ node{}};
    \node (u2) [right=1cm of u1] {$u_2$}
    child{ node{}}
    child{ node{}};
    \node (u3) [right=1cm of u2] {$u_3$}
    child{ node{}}
    child{ node{}};
    % \node (u4) [right=1cm of u3] {$u_4$}
    % child{ node{}}
    % child{ node{}};
    \node (u5) [right=1cm of u3] {$u'$}
    child{ node{}};

    \node[right=9.5cm of v]{$v$}
    child{ node{}};
    
    \node[right=8.0cm of ou,draw=white]{{\Huge $\Rightarrow$}};

    \node (u6) [right=4cm of u5] {$u_1$}
    child{ node{}}
    child{ node{}};
    \node (u7) [right=1cm of u6] {$u_2$}
    child{ node{}}
    child{ node{}};
    \node (u8) [right=1cm of u7] {$u_3$}
    child{ node{}}
    child{ node{}};
    \node (u9) [right=1cm of u8] {$u'$}
    child{ node{}}
    child{ node{}};
    % \node (u10) [right=1cm of u9] {$u'$}
    % child{ node{}}
    % child{ node{}};
  \end{tikzpicture}
}
  \caption{Example of procedure $\consolid(u,u')$. Left: nodes $u$ and $u'$ have one child. Right: node $u$ and its ancestors, up to a node $v$ that has two children, are removed. Children of $u_1$, $u_2$, $u_3$ are shifted one position to the right. Only relevant nodes and their children are shown.}
  \label{fig:consolid}
\end{figure}
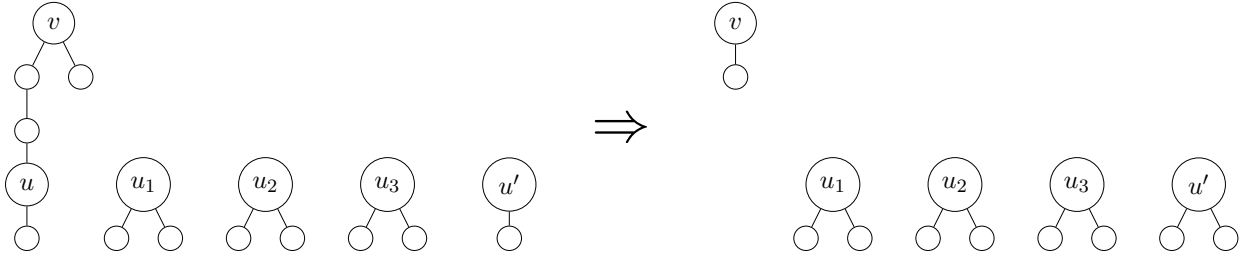

It remains to describe the procedure $\consolid(u,u')$. $\consolid(u,u')$ considers the children of nodes $u$, $u'$, and the children of all nodes between $u$ and $u'$. Every such node is moved by one position to the left. As a result, the node has no children and all other considered nodes have two children. Next, let $v$ be the lowest ancestor of $u$ that has two children. The node $u$ and all its ancestors that are below $v$ have no leaf descendants now. We remove the node $u$ and all nodes between $v$ and $u$. 
Now the node $v$ is a $1$-node. If $v$ is the root node, then we remove $v$. Otherwise, we check whether  $v$ has a neighbor $v'$, such that the distance between $v$ and $v'$ does not exceed $k_0+1$ and  $v'$ is a $1$-node. If $v'$ exists, we recursively call the procedure $\consolid(v,v')$ (respectively ($\consolid(v',v)$). There is at most one recursive call of our procedure  per tree level. Our procedure shifts $O(k_0)$ nodes by one position to the right and recursively calls itself on some higher tree level; every time when some node in $\Ts$ is shifted, we may have to move some leaf $\eps_i$ of $T$. Hence the total time of $\consolid(u,u')$ is $O(k_0\log^2 n)=O(\log^3 n)$. 
\begin{lemma}
  \label{lemma:worst1}
 We can implement a binary search tree so that  access to an element and an insertion of a new element  are supported in $O(\log^2 n)$ time. If an element $e_i$ was accessed $w_i$ times over a sequence of $W$ operations, then the depth of the leaf holding  $e_i$  does not exceed $\min(\log(W/w_i), \log n)+O(1)$.
\end{lemma}

\subsection{Fast Updates}
Now we show how the data structure from Section~\ref{sec:updates} can be changed to support updates in worst-case time. As in Section~\ref{sec:updates} the tree $\Ts$ is divided into the macro-tree and mini-trees. Each mini-tree contains $O(\log^3 n)$ pseudo-leaves. 

A new pseudo-leaf is inserted into a mini-tree; the cost of an insertion is  $O((\log\log n)^2)$ time by Lemma~\ref{lemma:worst1}. We run an additional background process that maintains the sizes of mini-trees. During each iteration we identify the largest mini-tree $T_l$  among all subtrees of size at least $(7/4)\log^3 n$.  We split $T_l$  into two mini-trees of almost-equal size.  We also identify the smallest mini-tree $T_k$ of size at most $(3/4)\log^3 n$; we merge $T_k$ with one of its direct neighbors (i.e., with the mini-tree immediately to the left or immediately to the right of $T_k$). If the resulting mini-tree  is larger than , then we split it into two almost-equal parts.  We will show in Section~\ref{sec:minisplit} how a mini-tree  can be split into two almost-equal parts or merged with another mini-tree in less than $O(\log^2 n)$ time.  When we split or merge two mini-trees, we also have to perform $O(1)$ updates on the macro-tree. The cost of updates is $O(\log^2 n)$, hence each iteration takes $O(\log^2 n(\log\log n)^2)$ time. By Theorem 5 from~\cite{DietzS87}, we can organize  our background process so that each mini-tree has no more than $2\log^3 n$ and no less than $\log^3 n/2$ pseudo-leaves.

\begin{lemma}
  \label{lemma:res2worst}
  We can implement a binary search tree so that  access to an element and an insertion of a new element  are supported in $O((\log \log n)^2)$ amortized time. If an element $e_i$ was accessed $w_i$ times over a sequence of $W$ operations, then the depth of the leaf holding  $e_i$  does not exceed $\min(\log(W/w_i), \log n)+O(1)$.
\end{lemma}
We can recursively apply Lemma~\ref{lemma:res2worst} in the same way as described in Section~\ref{sec:updates}.  To obtain the main result of this paper with worst-case guarantees, we apply Lemma~\ref{lemma:res2worst} $k+1$ times for a parameter $k>1$. 
\begin{theorem}
  \label{theor:mainworst}
  For any $k\ge 1$ there exists a binary search tree $T^{k}$, such that the depth of a leaf holding an element $e_i$ in $T^{k}$ does not exceed $\min(\log(W/w_i),\log n) +O(k)$ and the cost of updating $T^{k}$ after an element access or an insertion is $O(\log^{(k)}n+k)$.
\end{theorem}

\bibliographystyle{plainurl}
\bibliography{trees,2018_19_RGC}
\nocite{KarpinskiN09}

\section{Appendix}
\subsection{Alphabetic Codes}
\label{sec:alphabetic}
Our solution of the dynamic alphabetic tree problem can be also used to maintain a dynamic alphabetic code with almost-optimal length.  There is a one-to-one correspondence between binary prefix-free codes and binary trees. Given a tree $T$, every root-to-node path in $T$ can be encoded by a codeword: we write a $0$-bit for every edge from a node to its left child and a $1$-bit for every edge from a node to its right child. Using this correspondence we can translate our alphabetic tree method into an adaptive alphabetic encoding.
However this method would produce an alphabetic encoding in $O(\log n)$ time per symbol in the worst-case or in $O(1)$ time on average:  to generate a codeword we would need to traverse a path from the root to a leaf.  
In this section we show that an even better result is possible: we can  encode a sequence of symbols in $O(1)$ time per symbol. To achieve this goal,  some changes in the core method of maintaining the tree $\Ts$ are necessary. In our new method, described in Lemma~\ref{lemma:code1} we move only the leaves of $\Ts$.  This modification allows us to store explicitly the codewords all symbols. We will use the following notation in this section:  $w_i(t)$ will denote the number of times the symbol $a_i$ occurs in the length-$t$ prefix of a sequence $S$ and we will denote by  $S[t]$  the $t$-th symbol in $S$. 
\begin{lemma}
\label{lemma:code1}
  There is  a binary alphabetic code such that  the codeword length of the  symbol $S[t+1]=a_i$ does not exceed $\min(\log(t/w_i(t)),\log n) +O(1)$.  The cost of updating the code after encoding a symbol or inserting a new symbol into $\cC$ is  $O(\log^3 n)$. 
\end{lemma}
\begin{proof}
  We maintain a tree $\Ts$ of pseudo-leaves, as in the case of alphabetic trees. However, we use a different method to maintain the tree:   we only move the leaf nodes in the tree; internal nodes are never moved. There are $4W'$ pseudo-leaves in our modified tree and all pseudo-leaves are stored in an array $B$. Some pseudo-leaves are associated to symbols; we associate $2w'_i(t)$ pseudo-leaves to the symbol $a_i$, where $w'_i(t)=w_i(t)/\tau$ and $\tau=t/n$. The non-associated pseudo-leaves are called dummy pseudo-leaves. For any $i$, pseudo-leaves for the symbol $a_i$ precede the pseudo-leaves for the symbol $a_{i+1}$.  When a new pseudo-leaf is inserted, positions of some other pseudo-leaves in $B$ may change; that is, some pseudo-leaves are shifted to the right (or to the left) in $B$, but their relative order remains unchanged.   We can maintain elements in $B$  in such a way that at most $\log^2n$ pseudo-leaves are moved to different positions in $B$ after every insertion~\cite{Willard92, BenderFGKM17}.  For every $a_i$, there is a node $\eps_i$  such that all pseudo-leaves below $\eps_i$ are associated to $a_i$. The height of $\eps_i$ in $\Ts$ is $\log(w'_i(t)$.  The depth of $\eps_i$ in $\Ts$ is at most $\min(\log(t/w_i(t)),\log n) +O(1)$.   When a new element is inserted or when $w'_i(t)$ is changed, we need to insert a new pseudo-leaf and move $O(\log^2 n)$ other pseudo-leaves. After every move, we may have to update some node $\eps_i$; in this case we also re-compute the codeword for the symbol $a_i$.  For every moved pseudo-leaf, we identify its ancestor $\eps_i$;  if necessary, we change the position of $\eps_i$ in $\Ts$ and  update the codeword for $a_i$ in $O(\log n)$ time.  Since we move $O(\log^2 n)$ leaves, we might change codewords for $O(\log^2 n)$ symbols.  The total time needed to insert a new pseudo-leaf is $O(\log^3 n)$. 
\end{proof}
Thus our dynamic search trees produce a dynamic  alphabetic code that can be updated in $O(1)$ time and encodes the sequence of $W$ symbols in  $W(H+O(1))$ bits, where $H$ denotes the  entropy of $W$.  We can reduce the time to maintain the code with the same technique as in the case of binary trees.

\begin{lemma}
\label{lemma:code2}
There is a binary alphabetic code $\cC^{2}$, such that  (1) for $t=0$, $\ldots$, $W-1$ the codeword length of a symbol  $S[t+1]=a_i$ in $\cC^{2}$ does not exceed $\min(\log(t/w_i(t)),\log n) +O(1)$ (2) the  cost of updating $\cC^{2}$ after encoding a symbol or inserting a new symbol into $\cC^{2}$  is $O((\log \log n)^3)$.
\end{lemma}
\begin{proof}
  We divide the tree $\Ts$ into the macro-tree and mini-trees. Every mini-tree has $\Theta(\log^4 n)$ pseudo-leaves. We keep codewords for all nodes $\eps_i$  in the macro-tree. We also keep codewords for the roots of all mini-trees.  Finally 
we also store  mini-codewords for nodes $\eps_j$ stored in mini-trees. The mini-codeword of a node $\eps_j$ in a mini-tree $\Ts_g$ encodes the path from the root of $\Ts_g$ to $\eps_j$. 
We maintain mini-codewords using Lemma~\ref{lemma:code1}. When the number of pseudo-leaves in a mini-tree equals $2\log^4 n$, we split the mini-tree into two mini-trees of equal size. Then we insert a new leaf into the macro-tree and update it in $O(\log^3 n)$ time.    The total time needed to update the macro-tree, codewords of macro-leaves, and codewords of nodes in the macro-tree is $O(\log^3 n)$. Hence the total amortized cost of inserting a pseudo-leaf is $O((\log\log n)^3)$. 

For every symbol $a_j$  we record the node $\eps_j$ that holds $a_j$; if $\eps_j$ is stored in a mini-tree $T_g$, then the node $\eps_j$ also contains a pointer to $T_g$. In order to produce the codeword for the next symbol $S[t+1]=a_i$, we look up the node $\eps_i$ holding $a_i$. If $\eps_i$ is in the macro-tree, we output its codeword. If $\eps_i$ is in a mini-tree $T_g$, the codeword of $a_i$ is obtained by concatenating the codeword of $T_g$ and the mini-codeword of $\eps_i$.  
\end{proof}

\begin{theorem}
\label{theor:code}
There is a binary alphabetic code $\cC$, such that  (1) for $t=0$, $\ldots$, $W-1$ the codeword length of a symbol  $S[t+1]=a_i$ in $\cC$ does not exceed $\min(\log(t/w_i(t)),\log n) +O(1)$ (2) the  cost of encoding the next symbol, updating $\cC$ after encoding a symbol or inserting a new symbol into $\cC$  is $O(1)$.
\end{theorem}
\begin{proof}
We apply the method of Lemma~\ref{lemma:code2} to mini-trees. Every mini-tree is divided into a macro-mini tree and micro-trees. Leaves of a macro-mini tree are roots of micro-trees; each micro-tree contains between $d/4$ and $d$ pseudo-leaves for $d=(\log\log n)^4$.  The macro-mini tree (with codewords corresponding to tree leaves and codewords for nodes in the macro-mini tree) can be  updated in $O(\log\log n)^3$ time. We insert a new leaf into the macro-mini tree when a micro-tree is split in two parts, i.e., when $d$ new pseudo-leaves are inserted into a micro-tree. Thus the total cost of maintaining the macro-tree and macro-mini trees is $O(1)$. When the first $t$ symbols are encoded, the depth of the  node $\eps_i$ that holds a symbol $a_i$ is bounded by $\min(\log(t/w_i(t)),\log n) +O(1)$. 

Finally  we observe that we can simulate a micro-tree with $d$ pseudo-leaves using a look-up table.  Pseudo-leaves in such a small tree are associated to at most $d$ symbols. The approximate weight of every symbol is bounded by $d$. Hence the total number of different trees is bounded by $d^{d}$. We construct a look-up table that stores all possible trees on up to $d$ pseudo-leaves. For every micro-tree we store the index of  the  micro-tree obtained after inserting a new pseudo-leaf associated to the $j$-th smallest element (or the index of a tree obtained by inserting the pseudo-leaf for some new element). For every micro-tree we also keep positions of all nodes $\eps_j$ in the micro-tree, where $\eps_j$ is the node that holds the $j$-th smallest symbol.  Thus we need to store a universal table with $d^{d+1}=o(n)$ entries. This table can be initialized in $o(n)$ time. Hence the total cost of inserting a pseudo-leaf is $O(1)$ and we can maintain the code in $O(1)$ time per symbol. 

A codeword for any symbol $a_i$ can be generated in $O(1)$ time. If the node $\eps_i$ that holds $a_i$ is  in a micro-tree, then the codeword for $a_i$ consists of three parts: encoding of a leaf in the macro-tree, encoding of a leaf in the macro-mini tree, and encoding of the node $\eps_i$ in the micro-tree.  If $\eps_i$ is in a macro-mini tree, then the codeword for $a_i$ consists of two parts: encoding of a leaf in the macro-tree and  encoding of the node $\eps_i$ in  the macro-mini tree. If $\eps_i$ is in the macro-tree, then the codeword for $a_i$ is the codeword for $\eps_i$.  Every codeword consists of up to three components and  all components of a codeword are stored in our data structure. Hence, we can generate the codeword for any symbol in $O(1)$ time. 
\end{proof}

\subsection{ Proof of Lemma \ref{lem:Dynamic Entropy}}
\label{subsec:Lemma 1 proof}
\begin{proof}
We need to show that $\sum_{t=1}^W  \log \frac t {\max\left(w^{(t-1)}_{a_t},\, 1\right)}  \le W\cdot H + 2W$.  In this sum $t$ assumes all integer values between $1$ and $W$ and  $w_{a_i}^{(\cdot)}$ assumes all integer values between $1$ and $w_i$.  Therefore we have 

\begin{multline}
  \sum_{t=1}^W  \log \frac t {\max\left(w^{(t-1)}_{a_t},\, 1\right)}=
\sum _{t=1}^W \log t -\sum_{j=1}^n\sum_{i=1}^{w_j}\log i\le \\
\le W\log W -\sum_{j=1}^nw_j(\log w_j -2)=  W\log W -\sum_{j=1}^nw_j\log w_j+ 2\sum_{j=1}^nw_j=H\cdot W + 2W
\end{multline}

The inequality in the second line follows from the fact that $\sum_{j=1}^x\log j\le x\log x$ and $\sum_{j=1}^x \log j\ge x\log x -2x$.  The former inequality is obvious. For completeness, we prove the latter.
$\sum_{j=1}^x \log j=\sum_{i=0}^l L_i$ where 
$$L_i=(\log(x/2^i)+\log((x/2^i)-1)+\ldots +\log((x/2^{i+1})+1)$$ and 
$l=\log x$. Every term in $L_i$ is larger than $\log(x/2^{i+1})$; there are $x/2^{i+1}$ terms in $L_i$. Hence $L_i\ge (x/2^{i+1})(\log x- (i+1))$ and 
$$\sum_{i=0}^l L_i \ge \left(\sum_{i=0}^l\frac{1}{2^{i+1}}\right)x\log x - \left(\sum_{i=0}^l \frac{i+1}{2^{i+1}}\right)x\ge 
x\log x - 2x$$
because $\sum_{i=0}^l \frac{i+1}{2^{i+1}}\le 2$. 
\end{proof}

\subsection{Lemma~\ref{lemma:res2worst}: Splitting and Merging Mini-Trees}
\label{sec:minisplit}
It remains to describe how to split a mini-tree $T_m$ in less than $O(\log^2 n)$ time. We consider  all nodes in  $T_m$ of height $h=2\log\log n$. Each subtree with the root in a height-$h$ node has at most $2\log^2 n$ pseudo-leaves. By definition, the height of the tree is bounded by $3\log\log n+O(1)$. Hence there are $f=O(\log n)$ nodes of height $h$. We insert the $\floor{f/2}$ leftmost nodes into a new $k$-neighbor tree $T'_1$. In other words, we regard the nodes of height $h$ as leaves and insert them one-by-one into a new tree $T'_1$. The tree $T_1$ is obtained from $T'_1$: suppose that  a height-$h$ node $u$ in $T_m$  corresponds to a leaf $u_l$ in $T'_1$; we append the subtree rooted at $u$ in $T_m$ to $u_l$. Every leaf is inserted into $T'_1$ in $O((\log\log n)^2)$ time and we can append a sub-tree  to the corresponding  leaf in $O(1)$ time. Hence the mini-tree $T_1$ is constructed in $O(\log n(\log \log n)^2)$ time. $T_1$ is a $k$-neighbor tree because it satisfies conditions (1)-(3) from Section~\ref{sec:prelim}: All leaves of $T'_1$ have the same depth and all sub-trees appended to leaves of $T'_1$ have the same height. Therefore all leaves of $T_1$ have the same depth. All nodes in $T'_1$ satisfy conditions (2) and (3) because $T'_1$ is a $k$-neighbor tree.  We append subtrees of $T_m$ to leaves of $T_1'$ in the same order as they are stored in $T_m$ and $T_m$ is the $k$-neighbor tree. Therefore  all nodes on the $h$ lowest levels of $T_1$ also satisfy conditions (2) and (3).   The tree $T_2$ is obtained from $\ceil{f/2}$ rightmost height-$h$ nodes of $T_m$ in the same way. We can merge two mini-trees using a symmetric procedure. 

\end{document}